\documentclass[]{llncs}
\usepackage{graphicx}

\usepackage[T1]{fontenc}
\usepackage[utf8]{inputenc}
\usepackage{amsmath}
\usepackage{amssymb}
\usepackage[colorlinks,allcolors=blue]{hyperref}
\usepackage{cleveref}
\usepackage{thm-restate}
\usepackage{comment}
\usepackage{xspace}
\usepackage{subcaption}

\usepackage{url}

\makeatletter
\AtBeginDocument{%
	\def\doi#1{\url{https://doi.org/#1}}}
\makeatother

\newcommand{\mnev}{Mn\"{e}v\xspace}

\newcommand{\ER}{\ensuremath{\exists\mathbb{R}}\xspace}
\newcommand{\Poly}{\ensuremath{\mathsf{P}}\xspace}
\newcommand{\NP}{\ensuremath{\mathsf{NP}}\xspace}
\newcommand{\ETR}{\textsf{ETR}\xspace}
\newcommand{\curved}{\text{bi-curved}\xspace}
\newcommand{\Curved}{\text{Bi-curved}\xspace}
\newcommand{\separable}{dif\-fer\-ence-sep\-a\-ra\-ble\xspace}
\newcommand{\Separable}{Dif\-fer\-ence-sep\-a\-ra\-ble\xspace}
\newcommand{\representable}{\text{computable}\xspace}
\newcommand{\Representable}{\text{Computable}\xspace}

\newcommand{\R}{\mathbb{R}}

\newcommand{\A}{\ensuremath{\mathcal{A}}\xspace}
\newcommand{\B}{\ensuremath{\mathcal{B}}\xspace}
\newcommand{\F}{\ensuremath{\mathcal{F}}\xspace}
\newcommand{\stretchability}{\textsc{Stretcha\-bility}\xspace}
\newcommand{\stretcha}{\textsc{$d$\nobreakdash-Stretcha\-bility}\xspace}

\newcommand{\PPHA}{PPHA\xspace}
\newcommand{\PH}{pseudo\-hyper\-plane\xspace}
\newcommand{\Recognition}{\textsc{Recog\-ni\-tion}\xspace}

\usepackage{xcolor}

\def\arxiv{}  

\begin{document}

\title{The Complexity of Recognizing Geometric Hypergraphs%
\ifdefined\arxiv
\else
\thanks{The full version of this paper can be found on arXiv: \href{https://arxiv.org/abs/2302.13597v2}{2302.13597v2}.}
\fi
}

\author{Daniel~Bertschinger\inst{1} \and
Nicolas~El~Maalouly\inst{1}%
\ifdefined\arxiv
\else
\orcidID{0000-0002-1037-0203}
\fi
\and
Linda~Kleist\inst{2}%
\ifdefined\arxiv
\else
\orcidID{0000-0002-3786-916X}
\fi
\and
Tillmann~Miltzow\inst{3}%
\ifdefined\arxiv
\else
\orcidID{0000-0003-4563-2864}
\fi
\and
Simon~Weber\inst{1}%
\ifdefined\arxiv
\else
\orcidID{0000-0003-1901-3621}
\fi
}

\authorrunning{D. Bertschinger et al.}
%
\institute{Department of Computer Science, ETH Zürich \\ \email{\{daniel.bertschinger,nicolas.elmaalouly,simon.weber\}@inf.ethz.ch}
\and
Department of Computer Science, TU Braunschweig\\
\email{kleist@ibr.cs.tu-bs.de}\\
\and
Department of Information and Computing Sciences, Utrecht University\\
\email{t.miltzow@uu.nl}}
\maketitle              
\begin{abstract}
As set systems, hypergraphs are omnipresent and have various representations ranging from Euler and Venn diagrams to contact representations.
In a \emph{geometric representation} of a hypergraph $H=(V,E)$, each vertex $v\in V$ is associated with a point $p_v\in \R^d$ and each hyperedge $e\in E$ is associated with a connected set $s_e\subset \R^d$ such that  $\{p_v\mid v\in V\}\cap s_e=\{p_v\mid v\in e\}$ for all $e\in E$. 
We say that a given hypergraph $H$ is \textit{representable} by some (infinite) family $\F$ of sets in ${\R^d}$, if there exist $P\subset \R^d$ and $S \subseteq \F$ such that $(P,S)$ is a geometric representation of~$H$.
For a family \F, we define \Recognition{}(\F) as the problem to determine if a given hypergraph is representable by \F.
It is known that the \Recognition{} problem is \ER-hard for halfspaces in $\R^d$.
We study the families of translates of balls and ellipsoids in $\R^d$, as well as of other convex sets, and show that their \Recognition{} problems are also \ER-complete. This means that these recognition problems are equivalent to deciding
whether a multivariate system of polynomial equations with integer coefficients has a real solution.

\keywords{Hypergraph \and geometric hypergraph \and recognition \and computational complexity \and convex \and ball \and ellipsoid \and halfplane \and halfspace.}
\end{abstract}
%
%
%
\section{Introduction}\label{sec:introduction}

As set systems, hypergraphs appear in various contexts, such as databases, clustering, and machine learning.
They are also known as \emph{range spaces} (in computational geometry) or \emph{voting games} (in social choice theory). A hypergraph can be represented in various ways, e.g., by a bipartite incidence graph, a simplicial representation (if the set system is closed under taking subsets), Euler or Venn diagrams etc.
%
%
%
%
Similar to classical graph drawing,
one can represent vertices by points and hyperedges by connected sets in $\mathbb R^d$ such that each set contains exactly the points of a hyperedge.
For the purposes of legibility, uniformity, or also for aesthetic reasons, it is desirable that these sets satisfy additional properties, e.g., being convex or having similar appearance such as being homothetic copies or even translates of each other. 

For an introductory example, suppose we are organizing a conference and have a list of accepted talks.
Clearly, each participant wants to quickly identify talks of their specific interest.
In order to create a good overview, we want to find a good representation. To this end, we label each talk by several tags, e.g., \texttt{hypergraphs}, \texttt{complexity theory}, \texttt{planar graphs}, \texttt{beyond planarity},   \texttt{straight-line drawing}, \texttt{crossing numbers}, etc.
Then, we create a representation, where each tag is represented by a unit disk (or another nice geometric object of our choice) containing points representing the talks that have this tag, see \Cref{fig:Geometric-Hypergraph} for an example.
In other words, we are interested in a geometric representation of the hypergraph where the vertex set is given by the talks and tags define the hyperedges. 

\begin{figure}[htb]
	\centering
	\includegraphics[page=3,scale=1]{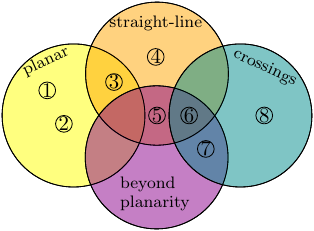}
	\caption{A geometric representation with unit disks of the abstract hypergraph~$H=(V,E)$ with $V=[8]$ and $E=\{\{1,2,3\},\{3,4,5,6\},\{5,6,7\},\{6,7,8\}\}$.}
	\label{fig:Geometric-Hypergraph}
\end{figure}

In this work, we investigate the complexity of deciding whether a given hypergraph has such a geometric representation.
We start with a formal definition.

\paragraph{Problem Definition.}
In a \emph{geometric representation} of a hypergraph $H=(V,E)$, each vertex $v\in V$ is associated with a point $p_v\in \R^d$ and each hyperedge $e\in E$ is associated with a connected set $s_e\subset \R^d$ such that  $\{p_v\mid v\in V\}\cap s_e=\{p_v\mid v\in e\}$ for all $e\in E$. 
We say that a given hypergraph $H$ is \textit{representable} by some (possibly infinite) family $\F$ of sets in ${\R^d}$, if there exist $P\subset \R^d$ and $S \subseteq \F$ such that $(P,S)$ is a geometric representation of~$H$.
For a family \F of geometric objects in $\mathbb R^d$, we define \Recognition{}(\F) as the problem to determine whether a given hypergraph is representable by \F. 
Next, we give some definitions describing
the geometric families studied in this work.

\paragraph{\Curved, \Separable, and \Representable  Convex Sets.}
We study convex sets that are \curved, \separable and \representable.
While the first two properties are needed for \ER-hardness,
the last one is used to show \ER-membership. 

Let $C\subset \R^d$ be a  convex set.
We call $C$ \textit{\representable} if
for any point $p\in \R^d$ we can decide in polynomial time on a real~RAM whether $p$ is contained 
in $C$.
We say that $C$ is \textit{\curved}
if there exists a unit vector $v\in\R^d$, such that there are two distinct tangent hyperplanes on $C$ with normal vector $v$;
with each of these hyperplanes intersecting $C$ in a single point, and $C$ being \emph{smooth} at both of these intersection points.
Informally, a convex set is \curved, if its boundary has two smoothly curved parts in which the tangent hyperplanes are parallel.
Note that a convex, \curved set is necessarily bounded.
As a matter of fact,  any strictly convex bounded set in any dimension is \curved. For such sets, any unit vector $v$ fulfills the conditions. As can be seen in \Cref{fig:niceShapeA}, being strictly convex is not necessary for being \curved.

\begin{figure}[htb]
	\centering
	\begin{subfigure}[t]{.32\textwidth}
	\centering
	\includegraphics[page=13]{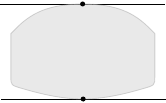}
	\caption{This burger-like set is \curved as shown by the two tangent hyper\-planes.
	}
	\label{fig:niceShapeA}
	\end{subfigure}\hfill
\begin{subfigure}[t]{.33\textwidth}
	\centering
	\includegraphics[page=14]{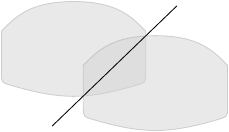}
	\caption{A hyperplane separating the symmetric difference of two translates of the burger-like set. 
	}
	\label{fig:niceShapeB}
\end{subfigure}\hfill
\begin{subfigure}[t]{.32\textwidth}
	\centering
	\includegraphics[page=15]{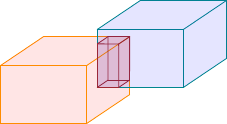}
	\caption{Two cubes in $\R^3$ whose symmetric difference cannot be separated by a plane.
	}
	\label{fig:niceShapeC}
\end{subfigure}
	\caption{Illustration for the notions \curved and \separable.}
\label{fig:niceShape}
\end{figure}

We call  $C$ \textit{\separable} if
for any two translates $C_1,C_2$ of $C$, there exists a hyperplane which strictly separates $C_1\setminus C_2$ from $C_2\setminus C_1$.
Being \separable is fulfilled by any convex set in $\R^2$, see \Cref{fig:niceShapeB} for an example. For a proof of this fact we refer to \cite[Corollary 2.1.2.2]{ma2000phd}. However, in higher dimensions this is not the case: 
for a counterexample, consider two $3$-cubes 
as in \Cref{fig:niceShapeC}. 
In higher dimensions, the \curved and \separable families include the balls and ellipsoids. 
We are not aware of other natural geometric families with those two properties. 
Note that balls and ellipsoids are naturally \representable.

We are now ready to state our results.

\subsection{Results}
We study the
recognition problem of geometric hypergraphs.
We first consider the maybe simplest type of 
geometric hypergraphs, namely those that stem from 
halfspaces.
It is known due to Tanenbaum, Good\-rich, and Scheinerman~\cite{tanenbaum1995Halfspaces} that the
\Recognition{} problem for geometric hypergraphs of halfspaces is \NP-hard, but their proof actually implies \ER-hardness as well. We present a slightly different proof of this fact due to two reasons. 
Firstly, their proof lacks details
about extensions to higher dimensions. 
Secondly, it is a good stepping stone towards our proof of \Cref{thm:translates}.

\begin{restatable}[Tanenbaum, Goodrich, Scheinerman~\cite{tanenbaum1995Halfspaces}]{theorem}{halfspaces}\label{thm:halfspaces}
	For every \mbox{$d\geq 2$}, \Recognition{}(\F) is \ER-complete for the family~\F of halfspaces in $\R^d$.
\end{restatable}
Next we consider families of objects that are translates of a given object.

\begin{restatable}{theorem}{translates}\label{thm:translates}
	For $d\geq2$, let $C\subseteq \R^d$ be a convex, \curved, \separable and \representable set,
	and let \F be the family of all translates of $C$.
	Then \Recognition{}(\F) is \ER-complete.
\end{restatable}
We note that for $d=1$, the \Recognition{} problems of halfspaces and translates of convex sets can be solved by sorting. Consequently, they can be decided in polynomial time.

One might be under the impression that the \Recognition{} problem
is \ER-complete for every reasonable family of geometric objects of dimension at least two.
However, we show that the problem is contained in \NP for translates of polygons and thus, if $\NP\subsetneq\ER$ as widely believed, not \ER-complete.

\begin{restatable}{theorem}{polygonTranslates}\label{thm:polygons}
	Let $P$ be a simple polygon with integer coordinates in $\mathbb R^2$,
	and \F the family of all translates of $P$.
	Then \Recognition{}(\F) is contained in \NP.
\end{restatable}

\paragraph{Organization.}
We give an overview over our proof techniques in \Cref{sec:sketch}. Full proofs of \Cref{thm:polygons} as well as the membership parts of \Cref{thm:halfspaces,thm:translates} are found in \Cref{sec:Membership}. We introduce the version of pseudohyperplane stretchability used in our hardness reductions in \Cref{sec:stretchability}. Full proofs of the hardness parts of \Cref{thm:halfspaces,thm:translates} can be found in \Cref{sec:halfspaces,sec:niceshapes}, respectively.

\subsection{Related work}

In this section we give a concise overview over related work on the complexity class \ER, geometric intersection graphs, and on other set systems related to hypergraphs.

\paragraph{The Existential Theory of the Reals.}
\label{par:ETR}
The complexity class \ER (pronounced as `ER' or `exists R')
is defined via its canonical complete problem \ETR (short for \emph{Existential Theory of the Reals}) and contains all problems that polynomial-time many-one reduce to it.
In an \ETR instance, we are given a sentence of the form
\[
\exists x_1, \ldots, x_n \in \R :
\varphi(x_1, \ldots, x_n),
\]
where~$\varphi$ is a well-formed and quantifier-free formula consisting of polynomial equations and inequalities in the variables and the logical connectives $\{\land, \lor, \lnot\}$.
The goal is to decide whether this sentence is true.

The complexity class \ER gains its importance from its numerous influential complete problems.
Important \ER-completeness results include the realizability of abstract order types~\cite{Mnev1988_UniversalityTheorem,Shor1991_Stretchability}, geometric linkages~\cite{Schaefer2013_Realizability}, and the recognition of geometric intersection graphs, as further discussed below.

More results concern graph drawing~\cite{dobbins_AreaUniversality_Journal,Erickson2019_CurveStraightening,Lubiw2018_DrawingInPolygonialRegion,Schaefer2021_FixedK}, the Hausdorff distance~\cite{HausDorff}, polytopes~\cite{Dobbins2019_NestedPolytopes,Richter1995_Polytopes}, Nash-equilibria~\cite{Berthelsen2019MultiPlayerNash,Bilo2016_Nash,Bilo2017_SymmetricNash,Garg2018_MultiPlayer,Schaefer2017_FixedPointsNash}, 
training neural networks~\cite{Abrahamsen2021NeuralNetworks,2022trainFull},
matrix factorization~\cite{Chistikov2016_Matrix,Schaefer2018_TensorRank,Shitov2016_MatrixFactorizations,Shitov2017_PSMatrixFactorization,tunccel2022computational}, 
continuous constraint satisfaction problems~\cite{Miltzow2022_ContinuousCSP},
geometric packing~\cite{Abrahamsen2020_Framework}, the art gallery problem~\cite{Abrahamsen2018_ArtGallery,stade2022complexity}, and covering polygons with convex polygons~\cite{Abrahamsen2022_Covering}.

\paragraph{Geometric Hypergraphs.}
Many aspects of hypergraphs with geometric representations have been studied. Hypergraphs represented by touching polygons in $\R^3$ have been studied by Evans et al.~\cite{evans}. Bounds on the number of hyperedges in hypergraphs representable by homothets of a fixed convex set $S$ have been established by Axenovich and Ueckerdt~\cite{axenovichHomothets}. Smorodinsky studied the chromatic number and the complexity of coloring of hypergraphs represented by various types of sets in the plane~\cite{smorodinsky}. Dey and Pach~\cite{deyPachExtremal} generalize many extremal properties of geometric graphs to hypergraphs where the hyperedges are induced simplices of some point set in $\R^d$. Haussler and Welzl~\cite{haussler1986epsilon} defined $\epsilon$-nets, subsets of vertices of hypergraphs called range spaces with nice properties. Such $\epsilon$-nets of geometric hypergraphs have been studied quite intensely~\cite{enets4,enets3,pachEpsilonNets2,pachEpsilonNets1}.

While there are many structural results, we are not aware of any research into the complexity of recognizing hypergraphs given by geometric representations, other than the recognition of embeddability of simplicial complexes, as we will discuss in the next paragraph.

\paragraph{Other Representations of Hypergraphs.}
Hypergraphs are in close relation with abstract simplicial complexes. In particular, an abstract simplicial complex (complex for short) is a set system that is closed under taking subsets. A $k$-complex is a complex in which the maximum size of a set is $k$. In a geometric representation of an abstract simplicial complex $H=(V,E)$ each $\ell$-set of $E$ is represented by a $\ell$-simplex such that two simplices of any two sets intersect exactly in the simplex defined by their intersection (and are disjoint in case of an empty intersection). Note that 1-complexes are graphs and hence deciding the representability in the plane corresponds to graph planarity (which is in \Poly). In stark contrast, Abrahamsen, Kleist and Miltzow recently showed that deciding whether a 2-complex has a geometric embedding in~$\mathbb R^3$ is \ER-complete~\cite{abrahamsenComplexes}; they also prove hardness for other dimensions.
Similarly, piecewise linear embeddings of simplicial complexes have been studied~\cite{vcadek2014computing,cadek2014polynomial,vcadek2017algorithmic,matouvsek2018embeddability,matousek2011hardness,mesmay2020embeddability,Skopenkov2020}.

\paragraph{Recognizing Geometric Intersection Graphs.}
Given a set of geometric objects,  its intersection graph has a vertex for each object, and an edge between any two intersecting objects.
The complexity of recognizing geometric intersection graphs has been studied for various geometric objects. We summarize these results in \Cref{fig:Class-Overview}. 

\begin{figure}[htb]
	\centering
	\includegraphics[scale=1]{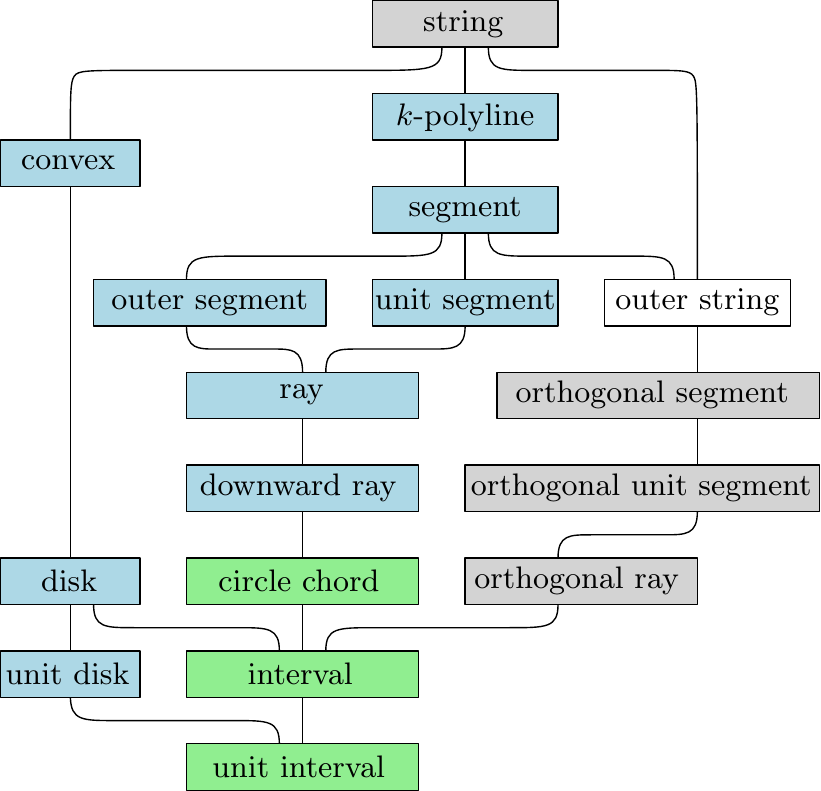}  
	\caption{Containment relations of geometric intersection graphs. 
		Recognition of a green class is in \Poly, of a grey class is \NP-complete, of a blue class is \ER-complete, and of a white class is unknown.}
	\label{fig:Class-Overview}
\end{figure}

While intersection graphs of 
circle chords (Spinnrad~\cite{spinrad1994recognition}),
unit intervals (Looges and Olariu~\cite{loogesUnitInterval}) 
and intervals (Booth and Lueker~\cite{boothInterval})
can be recognized in polynomial time, recognizing 
string graphs (Schaefer and Sedgwick~\cite{schaefer2003recognizing})
is NP-complete. In contrast, \ER-completeness of recognizing intersection graphs has been proved for 
(unit) disks by McDiarmid and M{\"u}ller \cite{mcdiarmid2013integer},
convex sets by Schaefer~\cite{schaeferConvex}, 
downward rays by Cardinal et al.~\cite{Cardinal2018_Intersection},
outer segments by Cardinal et al.~\cite{Cardinal2018_Intersection}, 
unit segments by Hoffmann et al.~\cite{unitSegER},
segments by Kratochv{\'\i}l and Matou{\v{s}}ek~\cite{Kratochvil1994_IntersectionGraphs},
$k$-polylines by Hoffmann et al.~\cite{unitSegER}, and
unit balls by Kang and M{\"u}ller~\cite{Kang2012}.

The existing research landscape indicates that recognition problems of intersection graphs are \ER-complete in case that 
the family of objects satisfy two conditions:
Firstly, they need to be ``geometrically solid'',
i.e., not strings.
Secondly, some non-linearity must be present by
either allowing rotations, or by the objects having some curvature.
Our results indicate that this general intuition might translate to the recognition of geometric hypergraphs.

\subsection{Overview of Proof Techniques}\label{sec:sketch}
We prove containment in \ER and \NP using standard arguments, providing witnesses and verification algorithms.

We prove the hardness parts of \Cref{thm:halfspaces,thm:translates} by reduction from stretchability of pseudohyperplane arrangements. 
The hypergraph we build from the given arrangement differs from the one built in the proof of \Cref{thm:halfspaces} given in~\cite{tanenbaum1995Halfspaces}, since we wish to use a single construction which works nicely for both theorems.
Given a simple pseudohyperplane arrangement $\A$, we construct a hypergraph~$H$ as follows: 
We double each pseudohyperplane by giving it a parallel \emph{twin}. In this arrangement, we place a point in every $d$-dimensional cell. These points represent the vertices of $H$. Every pseudohyperplane $\ell$ then defines a hyperedge, which contains all of the points on the same side of $\ell$ as its twin pseudohyperplane. See \Cref{fig:hypergraphConstruction} for an illustration of this construction.

Because this construction can also be performed on a hyperplane arrangement, it is straightforward to prove that if $\A$ is stretchable, $H$ can be represented by halfspaces. Conversely, we show that the hyperplanes bounding the halfspaces in a representation of $H$ must be a stretching of $\A$.

For \Cref{thm:translates}, \curved{}ness of a set $C$ implies that locally, $C$ can approximate any halfspace with normal vector close to $v$ as in the definition of \curved. This allows us to prove that stretchability of $\A$ implies representability of $H$ by translates of $C$. The set~$C$ being \separable is used when reconstructing a hyperplane arrangement from a representation of $H$.

\section{Membership}
\label{sec:Membership}
In this section we show \ER- and \NP-membership.

\subsection{Halfspaces}

For a given hypergraph $H$, it is not difficult to formulate an \ETR formula describing all needed properties for a geometric representation by halfspaces.
Therefore, we get the \ER-membership part of \Cref{thm:halfspaces}. 
\begin{restatable}{lemma}{memberHalfspaces} \label[lemma]{lemma:memberHalfspaces}
	Fix $d\geq 1$ and let  \F denote the family of halfspaces in $\R^d$. Then  \Recognition{}(\F) is contained in \ER.
\end{restatable}

The \ER-membership part of \Cref{thm:translates} is obtained by providing a simple verification algorithm~\cite{Erickson2022_SmoothingGap} (similar to how \NP-membership can be shown), based on the fact that our considered set $C$ is \representable.
\begin{restatable}{lemma}{memberTranslates}\label[lemma]{lemma:memberTranslates}
	For some $d\geq 1$, let $C\subseteq\R^d$ be a computable set and let \F be the family of all translates of $C$. Then, \Recognition{}(\F) is contained in \ER.
\end{restatable}

The full proofs of \Cref{lemma:memberHalfspaces,lemma:memberTranslates} have been omitted due to space constraints and can be found in 
\ifdefined\arxiv
\Cref{app:membership}.
\else
the full version of this paper.
\fi

\subsection{Translates of Polygons -- Proof of \Cref{thm:polygons}}

Here, we show \Cref{thm:polygons}, i.e., \NP-membership of \Recognition{} of translates of some simple polygon $P$.

\polygonTranslates*
\begin{proof}
	The proof uses a similar argument to the one used to show that
	the problem of packing translates of polygons inside a polygon is in \NP~\cite{Abrahamsen2020_Framework}.
	For an illustration, consider \Cref{fig:triangulated_polygon}.
	We first triangulate the convex hull of $P$, such that each edge of $P$ appears in the triangulation. 
	Then, a representation of a hypergraph $H$ by translates of $P$ gives rise to a certificate as follows: For each pair of a point $p$ and a translate $P'$ of $P$, we specify whether $p$ lies in the convex hull of $P'$. If it does, we specify in which triangle $p$ lies. Otherwise, we specify an edge of the convex hull for which $p$ and $P’$ lie on opposite sides of the line through the edge.
	
	\begin{figure}[htb]
		\centering
		\includegraphics[page =2]{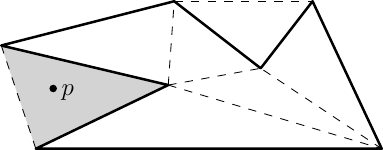}
		\caption{The polygon $P'$, a triangulation of its convex hull, and the triangle that contains $p$.}
		\label{fig:triangulated_polygon}
	\end{figure}

	Such a certificate can be tested in polynomial time: we create a linear program whose variables describe the locations of the points $p$ and the translation vectors of each translate of~$P$, and whose constraints enforce the points to lie in the areas described by the certificate. 
	This linear program has a number of constraints and variables polynomial in the size of $H$, and can be thus solved in polynomial time.
	
	The solution of this linear program gives the location of the points and the translation vectors of the polygons. This implies that these coordinates are all polynomial and could be used as a certificate directly.
 \qed
\end{proof}

\section{Pseudohyperplane Stretchability}\label{sec:stretchability}

A \emph{pseudohyperplane arrangement}
in $\R^d$ is an arrangement of 
\emph{pseudohyperplanes},
where a pseudohyperplane is a set homeomorphic to a hyperplane, and each intersection of pseudohyperplanes is homeomorphic to a plane of some dimension.
In the classical definition, every set of $d$ pseudohyperplanes has a non-empty intersection. Here, we consider \emph{partial pseudohyperplane arrangements (\PPHA{s})}, 
where not necessarily every set of $\leq d$ pseudohyperplanes has a common intersection.

A \PPHA is \emph{simple} if no more than $k$ pseudohyperplanes intersect in a space of dimension $d-k$, in particular, no $d+1$ pseudohyperplanes have a common intersection. We call the $0$-dimensional intersection points of $d$ pseudohyperplanes the \emph{vertices} of the arrangement.
A simple \PPHA \A is called \emph{stretchable} if there exists a hyperplane arrangement $\A'$ such that each vertex in \A also exists in $\A'$ and each (pseudo\nobreakdash-)hyperplane splits this set of vertices the same way in $\A$ and~$\A'$. In other words, each vertex of \A lies on the correct side of each hyperplane in~\A'.
We then call the hyperplane arrangement $\A'$ a \emph{stretching} of~\A.

The problem \stretcha is the problem of deciding whether a  simple \PPHA in $\R^d$ is  stretchable.
For $d=2$, \stretcha contains the stretchability of simple pseudoline arrangements which is known to be \ER-hard~\cite{Mnev1988_UniversalityTheorem,Shor1991_Stretchability}.
It is straightforward to prove \ER-hardness for all $d\geq 2$; the proof can be found in 
\ifdefined\arxiv
\Cref{app:stretcha}.
\else
the full version of this paper.
\fi

\begin{restatable}{theorem}{stretchability}\label{thm:hardnessofstretcha}
\stretcha is \ER-hard for all $d\geq 2$.
\end{restatable}

Similar extensions of pseudoline stretchability to higher dimensions have been studied in the literature. For example, \mnev{}'s universality theorem~\cite{Mnev1988_UniversalityTheorem} extends to higher dimensions, however we are not aware of any existing proofs that it also implies \ER-hardness in $d>2$. Kang and Müller~\cite{Kang2012} also studied a similar version of stretchability of partial arrangements of pseudohyperplanes.

\section{Hardness for Halfspaces -- Proof of \Cref{thm:halfspaces}
}
\label{sec:halfspaces}
We now present the hardness part of \Cref{thm:halfspaces}.
\halfspaces*
\begin{proof}[of \Cref{thm:halfspaces}]
	We reduce from \stretcha. Let $\A$ be a simple \PPHA.  For an example consider \Cref{fig:hypergraphConstructionA}. In a first step, we insert a parallel twin $\ell'$ for each \PH~$\ell$. The twin is close enough to $\ell$ such that $\ell$ and $\ell'$ have the same intersection pattern. Since $\ell$ and $\ell'$ are parallel, they do not intersect each other. This yields an arrangement $\A'$, see \Cref{fig:hypergraphConstructionB}.

	\begin{figure}[b!]
		\centering
		\begin{subfigure}[t]{.32\textwidth}
			\centering
			\includegraphics[page=1]{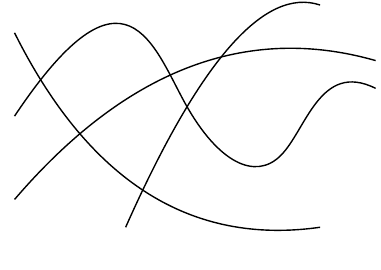}\hfill
			\caption{}
			\label{fig:hypergraphConstructionA}
		\end{subfigure}\hfill
		\begin{subfigure}[t]{.32\textwidth}
			\centering
			\includegraphics[page=3]{figures/halfspaces.pdf}\hfill
			\caption{}
			\label{fig:hypergraphConstructionB}
		\end{subfigure}\hfill
		\begin{subfigure}[t]{.32\textwidth}
			\centering
			\includegraphics[page=4]{figures/halfspaces.pdf}\hfill
			\caption{}
			\label{fig:hypergraphConstructionC}
		\end{subfigure}
		\caption{Construction of the hypergraph~$H$.
  (a)  A simple \PPHA $\A$. (b) The arrangement $\A'$ obtained by inserting twins. (c) The vertices of $H$ are the points in the cells, hyperedges of $H$ are defined by the pseudohalfspaces;
  the gray region shows one of the hyperedges.}
		\label{fig:hypergraphConstruction}
	\end{figure}

 	In a second step, we introduce a point in each $d$-dimensional cell of $\A'$; each point represents a vertex in our hypergraph $H$.
	Lastly, we define a hyperedge for each \PH $\ell$ of $\A'$: The hyperedge contains all of the points that lie on the side of $\ell$ that the twin \PH $\ell'$ lies in, see \Cref{fig:hypergraphConstructionC}.
	Note that we define such a hyperedge for every \PH of $\A'$. Thus, for every \PH~$\ell$ of the original arrangement $\A$ we define two hyperedges, whose union contains all vertices of $H$.
	
	It remains to show that $H$ is representable by halfspaces if and only if $\A$ is stretchable. 
	If $\A$ is stretchable, the construction of a representation of $H$ is straightforward: Consider a hyperplane arrangement \B which is a stretching of~$\A$.
	Then, for each hyperplane, we add a parallel hyperplane very close, so that their intersection patterns coincide. This results in a hyperplane arrangement $\B'$. We now prove that every $d$-dimensional cell of $\A'$ must also exist in $\B'$.
	First, note that each such cell corresponds to a cell of $\A$, which has at least one vertex on its boundary. All vertices of $\A$ exist in $\B$ by definition of a stretching. Furthermore, the subarrangement of the $d$ hyperplanes in $\B$ intersecting in this vertex must be simple, since their intersection could not be $0$-dimensional otherwise. In the twinned hyperplane arrangement $\B'$, all $3^d$ of the $d$-dimensional cells incident to this vertex (a cell is given by the following choice for each of the hyperplane pairs: above both hyperplanes, between the hyperplanes, or below both hyperplanes) must exist. This proves that all $d$\nobreakdash-dimensional cells of $\A'$ also exist in $\B'$. Inserting a point in each such $d$-dimensional cell and considering the (correct) halfspaces bounded by the hyperplanes of $\B'$ yields a representation of $H$.

	We now consider the reverse direction. Let $(P,\mathcal H)$ be a tuple of points and halfspaces representing $H$. 
	Let $h_{i,1}$ and $h_{i,2}$ be the two halfspaces associated with a \PH~$\ell_i$ of $\A$. 
	Let  $p_i$ denote the  $(d-1)$-dimensional hyperplane bounding $h_{i,1}$. We show that the family $\{p_i\}_i$ of these hyperplanes is a stretching of $\A$.
	
	For each intersection point $q$ of $d$ \PH{s} $\ell_1,\dots \ell_d$ in $\A$, we consider the  corresponding $2d$ \PH{s} in $\A'$. The \PPHA $\A'$ contains $3^d$ $d$-dimensional cells incident to their $2^d$ intersections; each of which contains a point. 
	We first show that the associated halfspaces must induce at least $3^d$ cells, one of which is bounded and represents the intersection point, see also \Cref{fig:HalfspacesStretchabilityA}:
	\begin{figure}[htb]
		\centering
		\begin{subfigure}[t]{.3\textwidth}
			\centering
			\includegraphics[page=16]{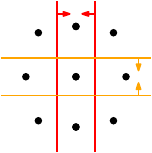}
			\caption{}
			\label{fig:HalfspacesStretchabilityA}
		\end{subfigure}\hfill
		\begin{subfigure}[t]{.3\textwidth}
			\centering
			\includegraphics[page=17]{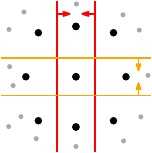}
			\caption{}
			\label{fig:HalfspacesStretchabilityB}
		\end{subfigure}\hfill
		\begin{subfigure}[t]{.3\textwidth}
			\centering
			\includegraphics[page=19]{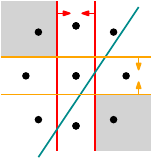}
			\caption{}
			\label{fig:HalfspacesStretchabilityC}
		\end{subfigure}
		\caption{Illustration for the proof of \Cref{thm:halfspaces} for $d=2$ showing that representability of $H$ implies stretchability of~\A. (a) Any two pseudolines $\ell_i,\ell_j$ in $\A$ have four corresponding lines bounding the respective halfplanes in $H$; these four lines induce 9 cells, each of which contains a point.  (b) Each point in $P$ belongs to exactly one of these 9 cells; the central bounded cell contains a unique point representing the intersection of $\ell_i$ and $\ell_j$. (c) The central bounded cell cannot be intersected by a line $p_k$ with $k\neq i,j$.}
		\label{fig:HalfspacesStretchability}
	\end{figure}
 	These $3^d$ points have pairwise distinct patterns of whether or not they are contained in each of the $2d$ halfspaces. Thus, these points need to lie in distinct cells of the arrangement of halfspaces, which proves the claim.
	Moreover, every point in $P$ belongs to exactly one of these $3^d$ cells. In particular, the central bounded cell, denoted by $c(q)$, contains exactly one point of~$P$, see \Cref{fig:HalfspacesStretchabilityB}.

	Now, we argue that the complete cell $c(q)$ (and thus in particular the intersection point of the hyperplanes representing $q$) lies on the correct side of each hyperplane $p$ in $\{p_i\}_i$. Note that, by construction of the hypergraph $H$, the $3^d$ points of $q$ lie on the same side of $p$. Suppose for a contradiction that $p$ intersects $c(q)$, see \Cref{fig:HalfspacesStretchabilityC}. Then there exist two unbounded cells incident to $c(q)$ which lie on different sides of $p$; these cells can be identified by translating $p$ until it intersects $c(q)$ only in the boundary. This yields a contradiction to the fact that the $3^d$ points of $q$ lie on the same side of $p$.

	We conclude that each intersection point of $d$ pseudohyperplanes in \A also exists in the arrangement $\{p_i\}_i$ and lies on the correct side of all hyperplanes. Thus, $\{p_i\}_i$ is a stretching of \A and \A is stretchable. \qed
\end{proof}

\section{Hardness for Convex, \Curved, and \Separable Sets  -- Proof of \Cref{thm:translates}}\label{sec:niceshapes}

We are now going to prove the hardness part of \Cref{thm:translates}.
\translates*

To this end, consider any fixed convex, \curved, and \separable set~$C$ in $\R^d$. Note that we can assume $C$ to be fully-dimensional, since otherwise each connected component would live in some lower-dimensional affine subspace, with no interaction between such components.
We use the same reduction from the problem \stretcha as in the proof for halfspaces in the previous section and show that the constructed hypergraph $H$ is representable by translates of $C$ if and only if the given \PPHA \A is stretchable.

\begin{restatable}{lemma}{approximateHalfspaces}\label[lemma]{lem:if}
	If \A is stretchable, $H$ is representable by translates of $C$.
\end{restatable}
The full proof of this lemma can be found in 
\ifdefined\arxiv
\Cref{app:approximateHalfspaces}. 
\else
the full version of this paper.
\fi
The idea behind the proof is that a stretching \A' of \A can be scaled  and stretched in such a way that every hyperplane has a normal vector close to the vector $v$ witnessing that $C$ is \curved, and such that all the vertices lie within some sufficiently small box. Then, for every halfspace $h_\pm$ bounded by some hyperplane $h$ in \A', there exists a translate of $C$ which approximates $h_\pm$ within the small box. This intuition is shown in \Cref{fig:largeness_intuition}. Since the hyperplane arrangement is simple, and there is some slack between the hyperplanes bounding the two twin halfspaces (as we argued above in the proof of \Cref{thm:halfspaces}), such an approximation is sufficient. 

\begin{figure}[htb]
    \centering
    \includegraphics[page=2,scale=1]{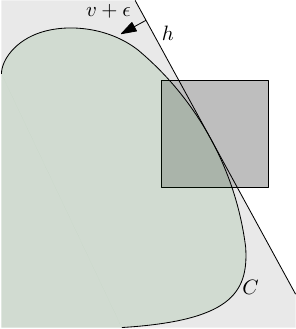}
    \caption{Illustration for the proof of \cref{lem:if}. Within the small box (dark grey), the translate of $C$ (green) approximates the halfspace (light grey) bounded by~$h$.}
    \label{fig:largeness_intuition}
\end{figure}

\begin{lemma}\label[lemma]{lem:onlyif}
	If the hypergraph $H$ is representable by translates of $C$, then 
	\A is stretchable.
\end{lemma}
\begin{proof}
	Assume $H$ is representable. 
	By construction, the two translates $C_{i,r},C_{i,l}$ of $C$ corresponding to the two hyperedges of each \PH $\ell_i$ must intersect as they contain at least one common point. We call their convex intersection the \emph{lens} of this \PH.
	For each \PH $\ell_i$ of~$\A$, we consider some hyperplane $p_i$ which separates $C_{i,r}\setminus C_{i,l}$ from $C_{i,l}\setminus C_{i,r}$. Such a hyperplane exists since $C$ is \separable.
	Let $\mathcal P:=\{p_i\}_i$  be the hyperplane arrangement consisting of all these separators. 
	We aim to show that $\mathcal P$ is a stretching of $\A$. 
	
	To this end, consider $d$ pseudohyperplanes $\ell_1,\ldots,\ell_{d}$ which intersect in $\A$. 
	\Cref{fig:stretcha} displays the case $d=2$.
	Furthermore, consider one more pseudohyperplane $\ell'$, and let $p'$, $C_r'$, $C_l'$ denote the separator hyperplane and translates of $C$ corresponding to $\ell'$. We show that the intersection $I_p:=p_1\cap \ldots\cap p_d$ is a single point which lies on the same side of $p'$ as the point $I_\ell:=\ell_1\cap\ldots\cap\ell_d$ lies of $\ell'$.

	\begin{figure}[htb]
		\centering
		\begin{subfigure}[t]{.48\textwidth}
			\centering
			\includegraphics[page=3]{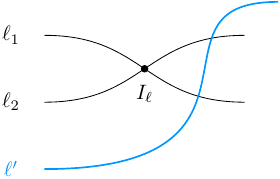}
			\caption{Pseudohyperplanes $\ell_1,\ell_2,\ell'$ in~$\A$}
			\label{fig:stretchaA}
		\end{subfigure}\hfill
		\begin{subfigure}[t]{.5\textwidth}
			\centering
			\includegraphics[page=4]{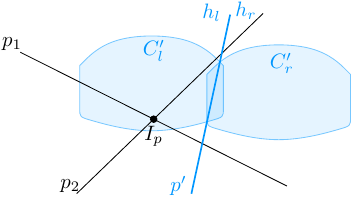}
			\caption{Hyperplanes $p_1,p_2,p'$ in $\mathcal P$.}
			\label{fig:stretchaB}
		\end{subfigure}
		\caption{Illustration for the proof of \Cref{lem:onlyif} for $d=2$. Some pseudohyperplanes in \A and their corresponding hyperplanes in $\mathcal P$.}
		\label{fig:stretcha}
	\end{figure}
	
	The hyperplane $p'$ divides the space into two halfspaces $h_r$ and $h_l$ such that $C_r' \backslash C_l'\subseteq h_r$ and $C_l' \backslash C_r'\subseteq h_l$.
	By construction,  the two hyperedges defined for $\ell'$ cover all vertices of $H$ and the vertices in the cells around $I_\ell$  belong to only one hyperedge.
	Suppose without loss of generality that these vertices only belong to the hyperedge represented by $C_l'$.
	We will show that the intersection $I_p$ must then be a point in $h_l$.
	
	We first show that the intersection $I_p$ is a point, i.e., $0$-dimensional. Consider all $2^d$ $d$-dimensional cells of \A around $I_\ell$. The construction of $H$ implies that each such cells contains a distinct point, and these points must all lie in distinct cells of the sub-arrangement of the involved hyperplanes $p_1,\ldots, p_d$. Assuming that $I_p$ is not a single point, this sub-arrangement is not simple, and the hyperplanes divide space into strictly fewer than $2^d$ cells, which results in a contradiction.

	Next we prove that $I_p$ is in $h_l$. Assume towards a contradiction that $I_p\in h_r$, see also \Cref{fig:stretcha2}.
 	Consider the $d$ lines that are formed by the intersections of subsets of $d-1$ hyperplanes among $p_1,...,p_{d}$. Each of these lines is the union of two rays beginning at $I_p$. Observe that the hyperplane $p'$ can only intersect one of the two rays forming each line. 
	Let $S$ be the convex cone centered at $I_p$ defined by the $d$ non-intersected rays. Observe that $S$ does not intersect $p'$, so $S$ must be fully contained in $h_r$, i.e., $S \cap h_l = \emptyset$. Note, however, by the construction of the hypergraph, there must be a point that lies in $S \cap (C_l'\setminus C_r')\subseteq S\cap h_l$, which is a contradiction.

 	\begin{figure}[htb]
		\centering
		\includegraphics[page=2]{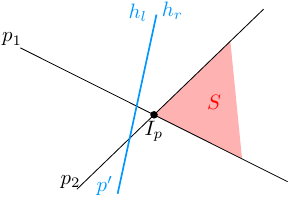}
		\caption{Illustration for the proof of \Cref{lem:onlyif} for $d=2$. The cone $S$ must intersect $C'_l\setminus C'_r$, which contradicts $I_p$ lying in $h_r$.
		}
		\label{fig:stretcha2}
	\end{figure}

	We conclude that $\mathcal P$ is a stretching of \A, and thus \A is stretchable. \qed
\end{proof}

\Cref{lem:if,lem:onlyif} combined now prove hardness of \Recognition{}(\F) for the family \F of translates of $C$. This completes the proof of \Cref{thm:translates}.

\section{Future directions}

We conclude with a list of interesting open problems:
As mentioned above, we are not aware of interesting
families of \curved and \separable sets in higher dimensions beyond balls and ellipsoids.
The families of translates of a given polygon show the
need for some curvature 
in order to show \ER-hardness.
We wonder if it is sufficient for \ER-hardness to assume 
curvature at only one boundary part instead of two opposite ones.
Another open question is to consider families that
include rotated copies or homothetic copies of a fixed geometric object.
Allowing for rotation, it is conceivable that \ER-hardness even holds for polygons.
Allowing for homothetic copies, the next natural question would be to study the family of all balls, not just balls of the same size.

\small
\bibliographystyle{splncs04}

\bibliography{literature,ETR}

\ifdefined\arxiv

\clearpage
\normalsize
\appendix

\section{Details of membership, \Cref{lemma:memberHalfspaces,lemma:memberTranslates}}\label{app:membership}
\memberHalfspaces*
\begin{proof}
	For a given hypergraph $H$, we formulate an \ETR formula as follows.
	For each vertex/point, we create variables $p = (p_1,\ldots,p_d)$ to represent the point.
	Similarly, for each hyperedge/halfspace, we create variables $h= (h_1,\ldots,h_{d+1})$ to represent the coefficients of the halfspace.
	Then for each point $p$ that is supposed 
	to be in some halfspace $h$, we create the constraint:
	\[h_1p_1 + \ldots h_dp_d \leq h_{d+1}.\]
	Similarly, if $p$ is not contained in a halfspace $h$, we create 
	the constraint:
	\[h_1p_1 + \ldots h_dp_d > h_{d+1}.\]
	This is a valid \ETR sentence that is equivalent to 
	the representability of $H$.
	Note that for any fixed dimension $d$ the \ETR sentence is 
	of polynomial size.
 \qed
\end{proof}

Recall that the class \NP is usually described by the existence of a witness and a verification algorithm. 
The same characterization exists for \ER using a real verification algorithm.
Instead of the witness consisting of binary words of polynomial length, 
in addition a polynomial number of real-valued numbers are allowed as a witness.
Furthermore, in order to be able to use those real numbers, the verification algorithm is 
allowed to work on the so-called real RAM model of computation.
The real RAM allows arithmetic operations with real numbers in constant time~\cite{Erickson2022_SmoothingGap}.

\memberTranslates*
\begin{proof}
	We describe a real verification algorithm as
	mentioned above. 
	The witness consists of the (real) coordinates of the points representing the vertices and the coefficients of the translation vectors 
	representing the hyperedges.
	By definition of \representable, a verification algorithm can efficiently check if each point is contained in the correct sets.
 \qed
\end{proof}

\section{Details of \Cref{thm:hardnessofstretcha}}\label{app:stretcha}

\stretchability*

\begin{proof}
	We reduce from stretchability of simple pseudoline arrangements, which is \ER-hard as shown in~\cite{Mnev1988_UniversalityTheorem,Shor1991_Stretchability}.
	
	Consider a simple pseudoline arrangement $L$ in the $x_1x_2$-plane. We consider $d-2$ pairwise orthogonal hyperplanes $h_1,\ldots,h_{d-2}$ whose common intersection is the $x_1x_2$-plane; e.g., the hyperplanes defined $x_i=0$ for $i=3,\dots, d$. The intersection of these hyperplanes serves as a canvas in which we aim to embed $L$.
	We extend each pseudoline of $\ell$ to a pseudohyperplane $h_\ell$ by extending it orthogonally to all $h_1,\ldots,h_{d-2}$, see \Cref{fig:CanvasHyperplane}.
	
	\begin{figure}[htb]
		\centering
		\includegraphics[scale=.8]{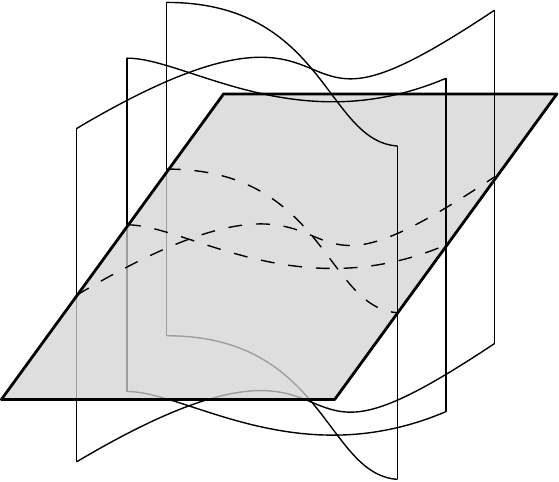}
		\caption{Extending a simple pseudoline arrangement (dashed) to a partial pseudohyperplane arrangement in $\R^3$. The grey hyperplane serves as the ``canvas''.}
		\label{fig:CanvasHyperplane}
	\end{figure}
	
	Clearly, the resulting pseudohyperplane arrangement \A can be built in polynomial time. Note that all intersection points of $d$ pseudohyperplanes in \A correspond to intersection points of $L$.
	
	If $L$ is stretchable, \A is clearly stretchable, as the above construction can be applied to the stretched line arrangement of $L$.
	
	If \A is stretchable, $L$ is stretchable, since restricting each hyperplane $h_\ell$ to the intersection of the hyperplanes $h_1,\ldots,h_{d-2}$ yields a line arrangement which is equivalent to $L$.
	
	As we have thus reduced stretchability of simple pseudoline arrangements to \mbox{\stretcha}, this concludes the proof.
 \qed
\end{proof}

\section{Details of \Cref{lem:if}}\label{app:approximateHalfspaces}
\approximateHalfspaces*
\begin{proof}
	We assume that \A is stretchable. We already proved in the previous section that thus there exists an arrangement of hyperplanes, in which we can create a twin of each hyperplane (with a tiny distance $\alpha$ between the twins), and in which we can place all the vertices of $H$ in the appropriate $d$-dimensional cells. If a vertex is placed between two twin hyperplanes, we assume it to be equidistant to them. As before, we denote this arrangement of hyperplanes and points by $\B'$.
	
	Let $v$ be the unit vector certifying that $C$ is \curved; recall the definition in \Cref{sec:introduction}. Because $C$ is smooth at the touching points of the tangent hyperplanes with normal vector~$v$, there exists $\epsilon>0$, such that any unit vector $w$ with $\lVert w-v \rVert_2\leq\varepsilon$ also fulfills the conditions to certify that $C$ is \curved.
	
	We now assume that $\B'$ fulfills the following properties: 
	\begin{enumerate}
		\item the normal vectors of all hyperplanes have distance at most $\varepsilon$ to $v$ or to $-v$
		\item every intersection point of $d$ hyperplanes as well as every point representing a vertex of~$H$, is contained in $[-1,1]^d$.
	\end{enumerate}
	Both properties can be achieved by applying some affine transformation with positive determinant, thus preserving the combinatorial structure of $\B'$.
	
	To represent the hyperedges of $H$, we will now use very large copies of $C$.
	Note that technically we are not allowed to scale $C$, but scaling $C$ by a factor $f$ is equivalent to scaling the arrangement by a factor $1/f$. Let $C^f$ be the set $C$ scaled by factor $f$.
	
	In order to determine the necessary scaling factor $f$, we consider the curvature of $C^f$ in all the points where the tangent hyperplanes of $C^f$ with normal vector~$w$ for~$\lVert w-v\rVert_2\leq \varepsilon$ intersect $C^f$. In each such tangent hyperplane $h$ with (unit) normal vector~$w$, we draw a $(d-1)$-ball $B$ of radius $10\sqrt{d}$ around the touching point $h\cap C^f$. Note that $10\sqrt{d}$ is larger than the length of any line segment contained in the box $[-1,1]^d$. Now,~$f$ has to be large enough such that $C^f$ contains every point $p+w\cdot \lambda$, for~$p\in B$ and $\alpha/10\leq \lambda\leq 10\sqrt{d}$.
	This ensures that the boundary of $C^f$ does not curve away from the tangent hyperplane too quickly, and that~$C^f$ is ``thick''. In other words, $C^f$ locally behaves like an only very slightly curved halfspace. See \Cref{fig:largeness} for an illustration of this requirement on $C^f$.
	
	\begin{figure}[htb]
		\centering
		\includegraphics{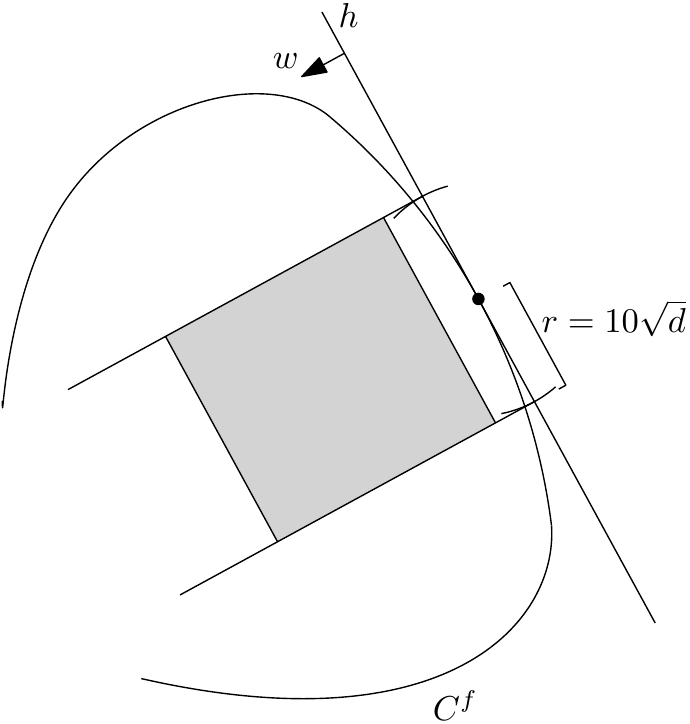}
		\caption{An illustration of the requirement on the scaling factor $f$. The set $C^f$ must contain the grey region.}
		\label{fig:largeness}
	\end{figure}
	
	We now replace each hyperplane $h$ of the arrangement $\B'$ by a translate $C_h^f$ of $C^f$, placed such that $h$ is a tangent hyperplane of $C_h^f$, the single point $h\cap C_h^f$ lies within the box $[-1,1]^d$, and $C_h^f$ lies completely to the side of $h$ containing its twin hyperplane. It remains to prove that $C_h^f$ contains exactly those points of $\B'$ which are on this side of $h$. Firstly, $C_h^f$ cannot contain more points, since $C_h^f$ is a subset of the halfspace delimited by $h$ containing its twin hyperplane. Second, we claim that $C_h^f$ contains all these points. To see this, note that within the box $[-1,1]^d$ containing all points, the boundary of $C_h^f$ is close enough to $h$ that it must contain all points between $h$ and its twin, since these points are located equidistant to the two hyperplanes. Furthermore, all points on the other side of the twin hyperplane are also contained in $C_h^f$ since within the box $[-1,1]^d$, the boundary $\delta(C_h^f)$ lies completely between $h$ and its twin hyperplane.
 \qed
\end{proof}

\fi
\end{document}